\documentclass[letterpaper, 10 pt, conference]{ieeeconf}  % Comment this line out if you need a4paper

\IEEEoverridecommandlockouts                              % This command is only needed if 
\usepackage{graphicx} % for pdf, bitmapped graphics files
\usepackage[linesnumbered,ruled]{algorithm2e}
\SetKwInput{KwInput}{Input}                % Set the Input
\SetKwInput{KwInit}{Initialize}                % Set the Input
\SetKwInput{KwSet}{Set}                % Set the Input
  \SetKwFunction{FMain}{Main}
  \SetKwFunction{FBreak}{Break$\_$Coalitions}
  \SetKwFunction{Prop}{Generate$\_$Proposal}
  \SetKwFunction{Free}{Update$\_$Free}
\SetKwBlock{Loop}{loop}{end}
\let\labelindent\relax
\usepackage{enumitem}
\usepackage{amsmath} % assumes amsmath package installed

\usepackage{amsthm} % assumes amsmath package installed

\usepackage{amssymb}  % assumes amsmath package installed

\newtheorem{definition}{Definition}
\newtheorem{proposition}{Proposition}
\newtheorem{assumption}{Assumption}

\newtheorem{thm}{Theorem}

\usepackage{nomencl}

\title{\LARGE \bf
Distributed Learning Dynamics for Coalitional Games
}

\author{Aya Hamed$^{1}$ and Jeff S. Shamma$^{1}$% <-this % stops a space
\thanks{*This work was supported by the University of Illinois Urbana-Champaign}% <-this % stops a space
\thanks{$^{1}$Aya Hamed and Jeff S. Shamma are with the Department of Industrial and Enterprise Systems Engineering, the Grainger College of Engineering, University of Illinois at Urbana-Champaign, 117 Transportation Building, MC-238, 104 S. Mathews Ave. Urbana, IL 61801-3080, USA {\tt\small ayah2@illinois.edu, jshamma@illinois.edu}}%
}

\begin{document}
\makenomenclature
\maketitle
\thispagestyle{empty}
\pagestyle{empty}

%%%%%%%%%%%%%%%%%%%%%%%%%%%%%%%%%%%%%%%%%%%%%%%%%%%%%%%%%%%%%%%%%%%%%%%%%%%%%%%%
\begin{abstract}
In the framework of transferable utility coalitional games, a scoring (characteristic) function determines the value of any subset/coalition of agents. Agents decide on both which coalitions to form and the allocations of the values of the formed coalitions among their members. An important concept in coalitional games is that of a core solution, which is a partitioning of agents into coalitions and an associated allocation to each agent under which no group of agents can get a higher allocation by forming an alternative coalition. We present distributed learning dynamics for coalitional games that converge to a core solution whenever one exists. In these dynamics, an agent maintains a state consisting of (i) an aspiration level for its allocation and (ii) the coalition, if any, to which it belongs. In each stage, a randomly activated agent proposes to form a new coalition and changes its aspiration based on the success or failure of its proposal. The coalition membership structure is changed, accordingly, whenever the proposal succeeds. Required communications are that: (i) agents in the proposed new coalition need to reveal their current aspirations to the proposing agent, and (ii) agents are informed if they are joining the proposed coalition or if their existing coalition is broken. The proposing agent computes the feasibility of forming the coalition. We show that the dynamics hit an absorbing state whenever a core solution is reached. We further illustrate the distributed learning dynamics on a multi-agent task allocation setting.

\end{abstract}

%%%%%%%%%%%%%%%%%%%%%%%%%%%%%%%%%%%%%%%%%%%%%%%%%%%%%%%%%%%%%%%%%%%%%%%%%%%%%%%%
\section{INTRODUCTION}
In strategic interactions with self-interested agents, cooperating with other agents can be the optimal strategic decision an agent can make to maximize its own benefits. In some interactive situations, the whole can be bigger than the sum of its parts, and agents forming coalitions can collectively gain higher payoffs. Coalitional games, within the field of cooperative game theory, enable us to study such situations. 

The authors in \cite{c1}\textendash\cite{c3} use coalitional game theory concepts to study the problem of forming coalitions of microgrids for local power exchange. Such collaboration between the microgrids increases their revenues and improves the autonomy of the system by decreasing their reliance on the main grid. In communications and networks fields, coalitional game theory offers a suitable framework to address a spectrum of problems. The tutorial paper \cite{c4} provides a classification of coalitional games and a variety of applications for each. One such application is the rate allocation problem presented in \cite{c5}, where the users of a multiaccess channel are modeled as agents in a coalitional game. The value of any coalition of agents depends on the maximum sum-rate achieved by the coalition given that the agents outside the coalition will act as jammers. The authors motivate various allocation methods for the agents including two of the most common solution concepts in coalitional game theory, the core and the Shapley value. 
 
Coalitional game theory is applied to various problems in multi-agent systems. Coalitional games framework was utilized in \cite{c6} in the control of a multi-agent system. Coalitions between agents dictate the communication between them and, hence, the collective optimization that they solve. The grand coalition, i.e., the union of all agents, represents a centralized optimization problem, and the partition of agents into singletons represents the fully decentralized optimization. The cost of coalitions depends on both the cost resulting from the coalition optimization as well as the cost of forming the coalition. Even though centralized optimization yields the least cost for the optimization problem, it has the highest overhead of communication between agents and complexity of the optimization problem. The authors in \cite{c7} and \cite{c8} use coalitional game theory solution concepts to solve the multi-robot task allocation problem. The Shapley value solution is used in \cite{c7} to group the robots into coalitions to fulfill the available tasks and specify the reward of each robot. The authors in \cite{c8} use coalition formation as an intermediary step to group robots and tasks into smaller coalitions, after which they solve for the optimal task allocation within these coalitions. 

In multi-agent systems, a centralized agent is not always accessible to mediate computing a solution with satisfactory payoffs for all agents. In addition, reliance on a centralized agent means having a single point of failure, which threatens the robustness of the system. Furthermore, privacy concerns and communication constraints may limit the information an agent is willing or capable of sharing. Our proposed distributed learning algorithm addresses these concerns. The presented algorithm converges to a core solution, a central solution concept in coalitional game theory, whenever one exists. The algorithm does not require the agents to learn about the full state of the environment. Alternatively, at any time instance, an agent is randomly activated, which then needs to know only the current payoff aspirations of a subset of its neighboring agents. Furthermore, the algorithm has small memory and computation requirements from the agents. They only need to retain their current payoff aspirations, coalition membership, and the values of the coalitions that they can form or a means of calculating such values online. 

The core of Transferable Utility (TU) coalitional games is defined by a distribution of the payoff resulting from forming the grand coalition. The core allocation respects the individual rationality of the agents such that no agent gains less in the grand coalition than what it can gain on its own. In addition, the core guarantees stability against group deviations where no group of agents can gain more by jointly deviating from the grand coalition. Finally, the core represents efficiency in that the total payoffs gained by the agents equals the value of the grand coalition. However, in some settings, the grand coalition is not the optimal coalition to form; specifically, when there is an associated cost with coalition formation that grows with the size of the coalition. The review paper \cite{c9} mentions the lack of literature in coalitional game theory for solution concepts that consider these settings, which readily occur in energy applications. In this paper, we consider a generalized core solution concept, similarly defined in \cite{c10} and \cite{c11}, which applies to such settings. This generalized core solution still preserves individual rationality and stability against deviation but with an extension to the efficiency concept such that the sum of the payoffs gained by the agents equals the maximum welfare that can be gained by any partitioning of the agents. Our proposed dynamics converge to this generalized core solution for general TU games, whenever one exists.

Classical core and Shapley value calculations do not consider the dynamics of coalition formation and dissolution. In addition to the assumption of the grand coalition optimality in most of the literature, the presented solutions and algorithms come short in connecting the reached solutions with real-life bargaining and negotiation setups. Our proposed dynamics exhibit feasible coalition formation dynamics throughout the horizon of iterations. Similar distributed aspiration-based algorithms were introduced in \cite{c12}\textendash \cite{c14}. Our algorithm builds upon the blind matching algorithm in \cite{c12} and \cite{c13}. It uses similar aspiration-based states and negotiation mechanism to reach the final allocations of the agents. The algorithm in \cite{c14} considers superadditive TU games, games where the value of a coalition is at least as good as the sum of values of any disjoint set of its subcoalitions. In superadditive TU games, the grand coalition is always optimal to form. The algorithm thus convergences to the standard core solution of these games. For general TU games, \cite{c10} and \cite{c11} provide similar convergence results to ours. However, they use best reply response, which imposes higher requirements on the agents' knowledge of their opponents' demands in addition to a global knowledge of the evolving coalition structure.

The rest of this paper is organized as follows. Section II presents the TU games setup and relevant propositions. Section III introduces the Coalition Proposal algorithm. Section IV discusses the convergence proof. Section V exhibits simulation results using the Coalition Proposal algorithm in a multi-agent task allocation setting. Finally, Section VI concludes the paper.

\section{ TRANSFERABLE UTILITY GAMES}

In this section, we present the TU games setup along with some definitions and propositions pertaining to it. We are following Hans Peters' book \cite{c15} in the fundamental definitions of the TU games, Definitions \ref{def1} and \ref{def2} below. However, we are generalizing the solution concept in Definition \ref{def3} to better suit general TU games where grand coalition formation does not necessarily result in the optimal welfare.

\begin{definition}
\label{def1}
A \textbf{transferable utility coalitional game} is defined by the pair $({N},v),$ such that ${N}=\{1,2,...,n\}$ is the set of players and $v:2^{{N}}\rightarrow \mathbb{R},$ where $v(\emptyset)=0,$ is the characteristic function defining the value of each coalition of players. An \textbf{allocation} $\mathbf{x}\in \mathbb{R}^n$ is a vector of real numbers representing the payoff distribution among the players.
\end{definition}

\begin{definition}
\label{def2}
For a TU game $(N,v),$ an allocation $\mathbf{x}$ is \textbf{indvidually rational} if $ x_i \geq v(\{i\})$ for all $i\in N$ and \textbf{coalitionally rational} if $ \sum_{i \in S}x_i \geq v(S)$ for all $S\subseteq N.$
 \end{definition}

\begin{definition} A collection $\rho$ of subsets of $N$ is defined as a \textbf{partition} of $N$ if $S\cap S'=\emptyset$ for all $S,S'\in\rho$ and $\cup_{S\in \rho}S=N.$ We denote the set of all partitions of $N$ by $\mathcal{P}(N).$

\end{definition}

\begin{definition}
\label{def3}
A \textbf{core solution} of a TU game $(N,v)$ is a pair $(\mathbf{x},\rho)$, where $\mathbf{x}$ is an allocation vector and $\rho \in \mathcal{P}(N)$ is a partition of players, satisfying that $\mathbf{x}$ is coalitionally rational, $ \sum_{i \in S}x_i \geq v(S)$ for all $S\subseteq N,$ and $\sum_{i\in {S}}x_i = v({S})\ $ for all ${S} \in \rho.$ We denote \textbf{the set of core solutions} of a game $(N,v)$ by $\Pi(v).$
\end{definition}

\begin{definition}
\label{def4}
    For a TU game $(N,v),$ we denote the \textbf{maximum welfare value}, $\max_{\rho\in\mathcal{P}(N)}\sum_{S\in\rho}v(S),$ by $K_v.$ 
\end{definition}

\begin{proposition}
\label{prop:prop1}
  For any TU game $({N},v),$ if $\Pi(v) \neq \emptyset,$ then $\sum_{i \in {N}}x_i=K_v\ $for all $ (\mathbf{x},\rho)\in \Pi(v)$.

\end{proposition}
\begin{proof}
Consider any $(\mathbf{x},\rho)\in \Pi(v).$ From the coalitional rationality, \begin{equation*}
    \sum_{i \in {S}}x_i \geq v({S})\ \forall {S} \subseteq {N}.
\end{equation*}
Then for any $\eta \in \mathcal{P}(N),$
\begin{equation*}
    \sum_{{S}\in \eta }\sum_{i \in {S}}x_i \geq \sum_{{S}\in \eta } v({S}),
\end{equation*}  and consequently, \begin{equation*}
    \sum_{i \in {N}}x_i\geq \sum_{{S}\in \eta } v({S}).
\end{equation*}
Since $\eta$ is an arbitrary partition of $N$, then \begin{equation}
\label{eqn1}
    \sum_{i \in {N}}x_i\geq \max_{\eta\in\mathcal{P}(N)}\sum_{S\in\eta}v(S).
\end{equation}
In addition, since $(\mathbf{x},\rho)\in \Pi(v),$  then, \begin{equation*}
    \sum_{i \in {S}}x_i = v({S})\ \forall {S} \in {\rho}.
\end{equation*}
Thus, \begin{equation*}
    \sum_{{S}\in \rho }\sum_{i \in {S}}x_i = \sum_{{S}\in \rho } v({S}),
\end{equation*} and since $\rho \in \mathcal{P}(N)$, then, \begin{equation}
\label{eqn2}
    \sum_{i \in {N}}x_i = \sum_{{S}\in \rho } v({S}).
\end{equation}

From (\ref{eqn1}) and (\ref{eqn2}), $$\sum_{i \in {N}}x_i = \max_{\eta\in\mathcal{P}(N)}\sum_{S\in\eta}v(S)=K_v.$$
\end{proof}
 
\begin{proposition}
\label{prop:prop2} 
 If $(\mathbf{x},\rho)\in \Pi(v),\ \mathbf{y}\in \mathbb{R}^n$ such that $\mathbf{y}$ is coalitionally rational, $\sum_{i \in {S}}y_i \geq v({S})\ $ for all $ {S}\subseteq {N},$ and $\sum_{i \in {N}}y_i=K_v,$ then $ (\mathbf{y},\rho)\in \Pi(v).$
\end{proposition}

\begin{proof}
    Given that $\sum_{i\in N}y_i=K_v$ and $(x,\rho)\in \Pi(v),$ then $\rho \in \mathcal{P}(N)$ and 
    \begin{equation}
        \label{eqn3}
        \sum_{S\in \rho}\sum_{i \in S}y_i= \sum_{i \in N}y_i=K_v.
    \end{equation}
    In addition, from Proposition \ref{prop:prop1} and the fact that $(x,\rho)\in \Pi(v),$ we get that 
    \begin{equation}
    \label{eqn4}
        K_v=\sum_{i \in N}x_i=\sum_{S\in \rho}\sum_{i \in S}x_i= \sum_{S\in \rho} v(S).
    \end{equation}
    From (\ref{eqn3}) and (\ref{eqn4}), 
    \begin{equation}
    \label{eqn5}
        \sum_{S\in \rho}\sum_{i \in S}y_i=\sum_{S\in \rho}v(S).
    \end{equation}
    We are given that 
    \begin{equation}
    \label{eqn6}
      \sum_{i  \in S}y_i\geq v(S)\ \forall S \subseteq N;  
    \end{equation}
     then, for (\ref{eqn5}) to be satisfied, we must have 
    \begin{equation}
    \label{eqn7}
      \sum_{i \in S}y_i=v(S)\ \forall S\in \rho.  
    \end{equation}
    Finally, (\ref{eqn6}) and ({\ref{eqn7}}) indicate that $(\mathbf{y},\rho)\in \Pi(v).$
\end{proof}

\section{COALITION PROPOSAL ALGORITHM}
 Algorithm 1 presents a pseudo code for our proposed algorithm. Informally, the algorithm proceeds as follows:

\begin{enumerate}[label=\arabic*.]
    \item Players come with arbitrary initial aspirations. The aspirations can be set to the players' singleton coalitions valuations or any bigger value on a grid of width $\delta$, a chosen discretization value.
    \item The algorithm then iterates over the following steps:
    \begin{enumerate}[label=\roman*.]
        \item A player is activated uniformly at random, which in turn chooses, using a uniform distribution, a set of other players to propose forming a coalition with. 
    \item The proposing player asks for, and receives, the current aspirations of the other players in the proposed coalition. 
    \item If the total of the received aspirations in addition to the proposing player's own aspiration raised by $\delta,$ is less than or equal to the proposed coalition valuation, the coalition is formed, otherwise, the proposal fails.
    \begin{enumerate}[label=\alph*.]
        \item If a coalition is successfully formed, the proposing player increases its aspiration by $\delta,$ and all the previous coalitions that had any player from the new coalition are dissolved. 
        \item If a coalition proposal fails, the proposing player decreases its aspiration by $\delta$ if both it is not in any other non-singleton coalition and it is not already at its singleton coalition valuation. Otherwise, the player's aspiration stays the same.
    \end{enumerate} 
    \end{enumerate}
\end{enumerate}

 Note that once we start the initial aspirations values on the $\delta$-discretized grid, the aspirations values stay on the $\delta$-discretized grid since the change in any player's aspiration from one iteration to the next is only a multiple of $\delta$.

\begin{algorithm}[!ht]
     \DontPrintSemicolon
    \caption{Coalition Proposal Algorithm}

    \SetKwProg{Fn}{Function}{:}{\KwRet}
    \KwInput{the set of players, ${N},$ and the characteristic function, $v$.}
    \KwSet{$\delta$ to be a discertization value where the $v(S)$ values for all $S\subseteq N$ lie on the $\delta$-discretized grid.}

  \KwInit{for all $i \in N, $\linebreak
   the player aspiration $a_i\leftarrow a^{0}_i,$ an arbitrary value on the $\delta$-discretized grid such that $a^{0}_i\geq v(\{i\})$, and \linebreak the player coalition 
   $C_i\leftarrow\emptyset $. }
   
    \Loop
    {
    Activate a player $i\in {N}$ uniformly at random.\\
        Player $i$ randomly chooses $S \subseteq {N}\setminus \{i\}.$ \\
        Player $i$ proposes to form coalition $J= S\cup\{i\}.$\\

    \eIf{$\sum_{j \in J}a_j+\delta \leq v(J)$ }
    { 
    $a_i\leftarrow a_i+\delta $  \\
    \text{Break old coalitions of players in }${J}:$ \linebreak For all $j\in J,$ for all $k\in C_j, k\neq j,\ C_k\leftarrow \emptyset$   \\

    Form coalition $J$: \linebreak $C_j\leftarrow J $ for all $j\in J$
    
    } 
    {\If{$C_i=\emptyset$}
    { $a_i\leftarrow\max(v(\{i\}),a_i-\delta$)\\
     \If{$a_i=v(\{i\})$}
    { $C_i\leftarrow \{i\}$
    }
   
    }

    }
    
    }

\end{algorithm}

\section{ANALYSIS OF THE COALITION PROPOSAL ALGORITHM}
In this section, we define the environment state at each iteration of the algorithm and prove how the environment state converges to a core solution of the input game $(N,v),$ whenever one exists.

\begin{definition}
    The \textbf{environment state} $(\mathbf{a},\mathcal{C})$ at any iteration of the Coalition Proposal algorithm is the \textbf{vector of aspirations} $\mathbf{a}\in \mathbb{R}^n$ and the \textbf{coalition structure} $\mathcal{C},$ where $\mathcal{C}$ is the set of disjoint formed coalitions at this iteration, i.e. $\mathcal{C}=\{C_i: i\in N\},$ where $C_i$ is player $i$'s current coalition.  
\end{definition}

Note that singleton coalitions can only belong to $\mathcal{C}$ when $a_i=v(\{i\}).$ Hence, $\mathcal{C}$ can be a strict subset of a partition of $N$ if there exists any agent with an empty coalition state.

\begin{definition}
    A \textbf{feasible environment state} is a pair of an aspiration vector and a coalition structure $(\mathbf{a},\mathcal{C}),$ where $\mathbf{a}\in \mathbb{R}^n$ and $\mathcal{C}$ is a set of disjoint subsets of $N$ such that $a_i\geq v(\{i\})$ for all $ i \in N$ and $\sum_{i \in S}a_i \leq v(S)$ for all $S \in \mathcal{C}.$ We denote the set of all feasible environment states by $\Omega(v).$ 

    \end{definition}

    \begin{proposition}
    Following the Coalition Proposal algorithm for a game $(N,v)$, starting from any environment state $(\mathbf{a},{\mathcal{C}})\in \Omega(v),$ the environment state stays in $\Omega(v)$.
\end{proposition}
\begin{proof}
    At any iteration, if a state $(\mathbf{a},{\mathcal{C}})\in \Omega(v),$ then $\ a_i\geq v(\{i\})$ for all $ i \in N$ and $\sum_{i \in S}a_i \leq v(S)$ for all $S \in \mathcal{C}.$ In the next iteration, only the proposing player changes its aspiration. The proposing player can only increase its aspiration if $\sum_{i \in S}a_i +\delta \leq v(S)$ is satisfied for its proposed coalition $S$, which can possibly be the same as its old coalition. All coalitions that do not include the proposing player are either dissolved or have no change in their members' aspirations. Hence, the next iteration's state, $(\mathbf{a}^+,{\mathcal{C}}^+)$, still satisfies $\sum_{i \in S}a^+_i \leq v(S)$ for all $S \in \mathcal{C}^+.$ Furthermore, if the proposal fails, the proposing agent cannot decrease its aspiration beyond its singleton coalition value and the rest of the agents do not change their aspirations, therefore, $a^+_i$ remains greater than or equal to $ v(\{i\})$ for all $i\in N$. Accordingly, if the environment is in a feasible environment state at one iteration, then, following the Coalition Proposal algorithm, the environment stays in a feasible environment state for the next iteration. 
\end{proof}

\subsection{Convergence Cases}
We divide the TU games into two cases. The first case is when the game $(N,v)$ has an empty core, $\Pi(v)=\emptyset$, and the second is when the game $(N,v)$ has a non-empty core, $\Pi(v)\neq\emptyset$.

\begin{assumption}
\label{asm:asm1}
    The discretization value, $\delta$, is chosen such that all the values of $v(S)$ for all $ S\subseteq N$ lie on the $\delta$-discretized grid. When $\Pi(v)\neq \emptyset,$ the choice of $\delta$ must also guarantee that there exists a state $(\mathbf{x},\rho)\in \Pi(v)$ where all the values of $\mathbf{x}$ lie on the $\delta$-discretized grid.

\end{assumption}
\subsubsection{Games with empty core $\Pi(v)=\emptyset$ }
\begin{proposition}
    Following the Coalition Proposal algorithm for a game $(N,v)$ with $\Pi(v)=\emptyset,$ under Assumption \ref{asm:asm1}, and starting from any environment state $(\mathbf{a},{\mathcal{C}})\in \Omega(v),$ the environment state never converges.

    \label{prop:Empty}
\end{proposition} 
\begin{proof}
    Let  $\mathcal{C'}={\mathcal{C}} \cup \{\{j\}:C_j=\emptyset$ and $a_j=v(\{j\})\}.$ 
    
    First, consider the instances where $ {\mathcal{C'}}\in \mathcal{P}(N),\ \mathcal{C'}$ is a partition of $N.$ Since $ {\mathcal{C'}}\in \mathcal{P}(N),$ then $\sum_{j \in {S}}a_j \leq v(S)$ for all $S\in \mathcal{C'}.$ Yet, $(\mathbf{a},{\mathcal{C'}})\notin \Pi(v),$ then there exists a set $ {S}\subseteq N$ such that $\sum_{j\in {S}}a_j +\delta \leq v({S}).$ Hence, if any player $i\in {S}$ is randomly activated, there is a positive probability that $i$ will propose ${S}$ to form and hence, player $i$ will increase its aspiration by $\delta.$
    
Second, consider the remaining instances where ${\mathcal{C'}}\notin \mathcal{P}(N).$ Then, there exists a player $ i \in {N}$ such that $C_i=\emptyset$ and $a_i \geq v(\{i\})+\delta.$ If any such player $i$ is activated, then for any choice of a proposed coalition, $a_i$ will either increase or decrease its aspiration by $\delta,$ depending on the success or failure of the proposal. Consequently, for all instances of the game state $(\mathbf{a},{\mathcal{C}}),$ there is a positive probability of generating a proposal that would lead to a change in a player's aspiration.

\end{proof}
\subsubsection{Games with non-empty core $\Pi(v)\neq \emptyset$ }

\begin{thm}
    Following the Coalition Proposal algorithm for a game $(N,v)$ with $\Pi(v)\neq\emptyset,$ under Assumption \ref{asm:asm1}, and starting from any environment state $(\mathbf{a},{\mathcal{C}})\in \Omega(v),$ the environment state converges to some state $(\mathbf{x},\rho)\in\Pi(v)$ with probability one.

    \label{thm1}
\end{thm} 
The next two subsections establish the necessary propositions and arguments to prove this theorem. Henceforth in this section, we will only discuss the games where $\Pi(v)\neq \emptyset.$

\subsection{Feasible States Partition}
\begin{definition}
Given a game $(N,v),$ we define the following subsets of the feasible states, $\Omega(v).$
\begin{itemize}
    \item $\Gamma(v)=  \{(\mathbf{a},{\mathcal{C}}) \in \Omega(v):\sum_{i \in {N}}a_i=K_v,\\ \sum_{i\in S}a_i \geq v(S) \text{ for all } {S}\subseteq {N}, \text{ and }{\mathcal{C}\notin \mathcal{P}(N)}\}$
    \item[]
    \item $\Upsilon(v)=\{(\mathbf{a},{\mathcal{C}}) \in \Omega(v):\sum_{i \in {N}}a_i>K_v,\\ \text{ and } \sum_{i\in S}a_i\geq v(S) \text{ for all } {S}\subseteq {N}\} $
    \item[]
    \item $\Psi(v)=\{(\mathbf{a},{\mathcal{C}}) \in \Omega(v):\exists {S}\subseteq {N} \\ \text{ such that }  \sum_{i\in S}a_i < v(S)\}$ 

\end{itemize}
\end{definition}

\begin{proposition}
\label{prop:prop5}
    The set of core solutions $\Pi(v)$ is a subset of the feasible environment state $\Omega(v).$
\end{proposition}
\begin{proof}
    Consider any $(\mathbf{x},\rho)\in \Pi(v).$ From coalitional rationality, $x_i\geq v(\{i\})$ for all $i\in N.$ In addition, $\rho \in \mathcal{P}(N)$ and $\sum_{i\in S}x_i =v(S)$ for all $S\in\rho.$ Hence, $(\mathbf{x},\rho)$ satisfies all the conditions of a feasible environment state and, thus, $(\mathbf{x},\rho)\in \Omega(v)$.
\end{proof}

\begin{proposition}
\label{prop:prop6}
    If $\Pi(v) \neq \emptyset, $ then for any $(\mathbf{a},{\mathcal{C}}) \in \Omega(v)$, if $\sum_{i\in S}a_i \geq v(S)\ $ for all $ S \subseteq N,$ then $\sum_{i \in N}a_i\geq K_v.$
\end{proposition}
\begin{proof}
    Consider any $(\mathbf{a},{\mathcal{C}}) \in \Omega(v)$ such that $\sum_{i\in S}a_i\geq v(S) \text{ for all } {S}\subseteq {N}$ and any $(\mathbf{x},\rho) \in \Pi(v).$ Granted, \[v({S})\leq \sum_{i \in {S}}a_i\ \forall S \in \rho.\] Yet, from the definition of  $\Pi(v)$,  \[v({S}) = \sum_{i \in {S}}x_i\ \forall S \in \rho.\] 
    Hence,\[\sum_{i \in {S}}x_i\leq \sum_{i \in {S}}a_i\ \forall S \in \rho.\]
    From Proposition \ref{prop:prop1}, \[\sum_{i \in {N}}x_i=K_v.\]  As a result, \[K_v=\sum_{{S}\in \rho}\sum_{i \in {S}}x_i\leq \sum_{{S}\in \rho}\sum_{i \in {S}}a_i= \sum_{i \in {N}}a_i.\]
\end{proof}
\begin{proposition}

    $\{\Pi(v),\Gamma(v),\Upsilon(v),\Psi(v)\}$ is a partition of $ \Omega(v)$.
\end{proposition}
\begin{proof}
    From the definitions of $\Pi(v),\Gamma(v),\Upsilon(v),\Psi(v)$ and Proposition \ref{prop:prop5}, the sets $\Pi(v),\Gamma(v),\Upsilon(v),\Psi(v)$ are all subsets of $\Omega(v).$ In addition, using the sets' definitions and Proposition \ref{prop:prop6}, each $(\mathbf{a},{\mathcal{C}}) \in \Omega(v)$ belongs to exactly one of the aforementioned sets as follows. If there exists a set $ {S}\subseteq {N}$ such that $\sum_{i \in S}a_i<v(S),$ then $(\mathbf{a},{\mathcal{C}})\in \Psi(v).$ If $\sum_{i\in S}a_i\geq v(S) \text{ for all } {S}\subseteq {N}$ and $\sum_{i \in {N}}a_i> K_v,$ then $(\mathbf{a},{\mathcal{C}})\in \Upsilon(v).$ If, however, $\sum_{i\in S}a_i\geq v(S) \text{ for all } {S}\subseteq {N}$ and $\sum_{i \in {N}}a_i = K_v,$
    then $(\mathbf{a},{\mathcal{C}})\in \Pi(v)$ if ${\mathcal{C}}\in \mathcal{P}(N)$ and $(\mathbf{a},{\mathcal{C}})\in \Gamma(v)$ if ${\mathcal{C}}\notin \mathcal{P}(N).$
\end{proof}
\subsection{The Steering Sequences}
In this subsection, we exhibit possible positive probability sequences of proposals that steer any state $(\mathbf{a},{\mathcal{C}}) \in \Omega(v)$ to some state $(\mathbf{x},\rho) \in \Pi(v).$ The following definitions and propositions are needed for constructing the aforementioned sequences. 

\begin{definition}
  For any $\mathbf{a},\mathbf{x} \in \mathbb{R}^n$, define the \textbf{lower-valued indices set} as $L_{\mathbf{a},\mathbf{x}}=\{i:a_i<x_i\},$ the  \textbf{upper-valued indices set} as $U_{\mathbf{a},\mathbf{x}}=\{i:a_i>x_i\},$ and the \textbf{equal-valued indices set} as $E_{\mathbf{a},\mathbf{x}}=\{i:a_i=x_i\}.$
\end{definition}

\begin{definition}
  For a given $(\mathbf{a},{\mathcal{C}}) \in \Omega(v),$ the \textbf{banded set of players} is ${B}({\mathcal{C}})=\{i: \exists {S}\in {\mathcal{C}}$ such that $ i\in {S}$\} and the \textbf{free set of players} is $F(\mathcal{C})={N}\setminus B(\mathcal{C})$.
\end{definition}
\begin{proposition}
\label{prop:prop7}
    Assuming $\Pi(v)\neq \emptyset,$ then for all $ (\mathbf{a},{\mathcal{C}}) \in \Psi(v)$ and $(\mathbf{x},\rho) \in \Pi(v)$, if ${S}\subseteq {N}$ is such that $\sum_{i\in S}a_i < v(S)$, then $L_{\mathbf{a},\mathbf{x}} \cap {S} \neq \emptyset.$

\end{proposition}

 Proposition \ref{prop:prop7} states that if the core is nonempty and the game is in a $\Psi(v)$ state, then given any core solution and any set ${S}$ such that $\sum_{i\in S}a_i < v(S),$ there is at least one player in ${S}$ that has an aspiration lower than its allocation in that core solution.

\begin{proof}
    We prove this proposition by contradiction. Assume $\exists (\mathbf{a},{\mathcal{C}})\in \Psi(v), (\mathbf{x},\rho)\in\Pi(v),$ and ${S}\in \{{S'}\subseteq {N}:\sum_{i\in S'}a_i < v(S') \} $ such that $L_{\mathbf{a},\mathbf{x}} \cap {S} = \emptyset.$ \\Since $L_{\mathbf{a},\mathbf{x}} \cap {S} = \emptyset,$ then, \[ {S}\subseteq U_{\mathbf{a},\mathbf{x}}\cup E_{\mathbf{a},\mathbf{x}},\] which in turn implies that \[a_i\geq x_i\ \forall i\in {S}.\] Therefore, \[\sum_{i \in{S}}a_i\geq \sum_{i \in{S}}x_i.\] Given that $(\mathbf{x},\rho)\in\Pi(v),$ we get that \[\sum_{i \in{S}}a_i\geq \sum_{i \in{S}}x_i\geq v({S)},\] and thus, $\sum_{i\in S}a_i\geq v(S),$ which is a \textbf{contradiction}.
\end{proof}
\begin{proposition}
\label{prop:prop8}

    If $\ \Pi(v)\neq \emptyset,$ then for all $ (\mathbf{a},{C}) \in \Upsilon(v)$ and $(\mathbf{x},\rho) \in \Pi(v)$, $U_{\mathbf{a},\mathbf{x}}\cap F(\mathcal{C}) \neq \emptyset.$

\end{proposition}
 Proposition \ref{prop:prop8} states that if the core is nonempty and the game is in a $\Upsilon(v)$ state, then given any core solution, there is at least one player that is free and has an aspiration higher than its allocation in the given core solution. 
  
\begin{proof}
    We prove this proposition by contradiction. Assume $\exists (\mathbf{a},{\mathcal{C}})\in \Upsilon(v)$ and $ (\mathbf{x},\rho)\in \Pi(v)$ such that $U_{\mathbf{a},\mathbf{x}}\cap F(\mathcal{C}) = \emptyset.$  Since $U_{\mathbf{a},\mathbf{x}} \cap F(\mathcal{C}) = \emptyset,$ then, \[ F(\mathcal{C})\subseteq L_{\mathbf{a},\mathbf{x}}\cup E_{\mathbf{a},\mathbf{x}},\] which in turn implies that \[a_i\leq x_i\ \forall i\in F(\mathcal{C}).\] Hence, 
    \begin{equation}
    \label{eqn:eqn7}
        \sum_{i \in F(\mathcal{C})}a_i\leq \sum_{i \in F(\mathcal{C})}x_i.
    \end{equation}
    Since $(\mathbf{a},{\mathcal{C}})\in\Omega(v),$ then \begin{equation*}
        \sum_{i\in B(\mathcal{C})}a_i=\sum_{{S}\in{\mathcal{C}}}\sum_{i\in{S}}a_i\leq\sum_{{S}\in{\mathcal{C}}}v({S}).
    \end{equation*} In addition, $(\mathbf{x},\rho)\in \Pi(v),$ which implies that
    \begin{equation*}
        \sum_{{S}\in{\mathcal{C}}}v({S})\leq \sum_{{S}\in{\mathcal{C}}}\sum_{i\in{S}}x_i .
    \end{equation*}
    As a result,
    \begin{equation}
    \label{eqn:eqn8}
        \sum_{i\in B(\mathcal{C})}a_i \leq  \sum_{i\in B(\mathcal{C})}x_i. 
    \end{equation}
    From (\ref{eqn:eqn7}), (\ref{eqn:eqn8}), and Proposition \ref{prop:prop1}, \[\sum_{i\in{N}}a_i \leq \sum_{i\in{N}}x_i=K_v,\] which \textbf{contradicts} the definition of $\Upsilon(v)$.
\end{proof}
  
The steering sequences are used to show that there is a positive probability of picking a sequence of proposals that lead any feasible state $(\mathbf{a},{\mathcal{C}})$ to a core solution of the game when the set of core solutions is non-empty. The steering sequence is constructed upon choosing and fixing any $(\mathbf{x}^*,\rho^*)\in \Pi(v)$ that lies on the $\delta$-discretized grid, which is assumed to exist by Assumption \ref{asm:asm1}. The sequence has two main stages. If $(\mathbf{a},{\mathcal{C}})\in \Upsilon(v)\cup \Psi(v),$ the sequence will follow the first stage proposals till it reaches $\Gamma(v)\cup \Pi(v).$ When the state reaches or starts in $\Gamma(v),$ the sequence will follow the second stage proposals.

\par
\textbf{First stage:} If the state $(\mathbf{a},{\mathcal{C}})\in \Upsilon(v)\cup \Psi(v),$ we choose a feasible proposal $(i,S)$ that transitions this state to a state $(\mathbf{a}^+,{\mathcal{C}}^+)$ with $\lVert \mathbf{a}^+-\mathbf{x}^* \rVert _1 = \lVert \mathbf{a}-\mathbf{x}^* \rVert _1-\delta$. Choosing $(i,S)$ means that player $i$ proposes to form $S.$ Specifically, the chosen proposals will be as follows:

\begin{itemize}
    \item If $(\mathbf{a},{\mathcal{C}})\in \Psi(v),$ pick ${S}\subseteq N$ such that $\sum_{i\in S}a_i < v(S).$ From Proposition \ref{prop:prop7}, $L_{\mathbf{a},\mathbf{x^*}} \cap {S} \neq \emptyset.$ Pick any $i\in L_{\mathbf{a},\mathbf{x}^*} \cap {S}.$ Then, \textbf{the proposal} ${(i,{S})}$ \textbf{is chosen.} ${S}$ will be successfully formed and only player $i$'s aspiration will increase by $\delta,$ $a^+_i=a_i+\delta$. Since $i\in L_{\mathbf{a},\mathbf{x^*}},$ $\lVert \mathbf{a}^+-\mathbf{x}^* \rVert _1 = \lVert \mathbf{a}-\mathbf{x}^* \rVert _1-\delta.$

    \item If $(\mathbf{a},{\mathcal{C}})\in \Upsilon(v),$ from Proposition \ref{prop:prop8}, $U_{\mathbf{a},\mathbf{x^*}} \cap F(\mathcal{C}) \neq \emptyset.$ Pick any $i\in U_{\mathbf{a},\mathbf{x^*}} \cap F(\mathcal{C})$ and any ${S}\subseteq {N}$ such that $i\in{S}.$ Then, \textbf{the proposal} ${(i,{S})}$ \textbf{is chosen.} ${S}$ will fail to form and player $i$ is free and has an aspiration $a_i \geq v(\{i\})+\delta.$ Hence, the proposal failure will cause the player's aspiration to decrease by $\delta,$ $a^+_i=a_i-\delta$. Since $i\in U_{\mathbf{a},\mathbf{x^*}},$ $\lVert \mathbf{a}^+-\mathbf{x}^* \rVert _1 = \lVert \mathbf{a}-\mathbf{x}^* \rVert _1-\delta.$

\end{itemize}

Note that if $(\mathbf{a},{\mathcal{C}})\in \Upsilon(v)\cup \Psi(v),$ the above proposals will consistently reduce the $1$-norm between the new state and $\mathbf{x^*}$ by $\delta.$ We can continue performing these proposals until the state reaches $\Gamma(v) \cup \Pi(v).$ This will happen within a finite number of steps, which is at most the $1$-norm, between the starting state and the chosen core solution, divided by $\delta.$

\textbf{Second stage: } If the state $(\mathbf{a},{\mathcal{C}})\in \Gamma(v),$ a sequence of proposals is selected to transition said state to a new state $(\mathbf{a},{\mathcal{C}'})\in \Pi(v).$

\begin{itemize}
    \item[] If $(\mathbf{a},{\mathcal{C}})\in \Gamma(v),$ then $F(\mathcal{C})\neq \emptyset.$ Pick any $i\in F(\mathcal{C})$ and pick the set ${S}\in \rho^*$ where $i\in{S}.$ \textbf{Choose the proposal} ${(i,{S})}.$ From Proposition \ref{prop:prop2}, $(\mathbf{a},{\rho^*})\in \Pi(v).$ Hence, $ \sum_{j\in S}a_j= v(S) \text{ for all } {S}\in \rho^*. $ Thus, the proposal will fail resulting in $a^+_i=a_i-\delta$ and $v(S)-\sum_{j\in S}a^+_j=\delta.$ Then \textbf{pick the proposal} ${(i,{S})}$ again. Now, the proposal will succeed and $a^{++}_i=a^+_i+\delta=a_i$ and $v(S)-\sum_{j\in S}a^{++}_j=0.$ Now, $\mathbf{a}^{++}=\mathbf{a}$, hence, the new state is either in $\Gamma(v)$ or ${\Pi}(v).$ If the state is still in $\Gamma(v),$ then make another iteration of the second stage.
\end{itemize}

Note that this stage ends in a finite number of steps. That is because if ${S}\in \rho^*$ is chosen as a part of a proposal in one iteration, none of the players in $ {S}$ become free again. Hence, the number of iterations of the above two proposals is bounded by the number of elements of $\rho^*.$

\textbf{The sequence achieves its target whenever the state reaches $\Pi(v).$} As demonstrated next, once a state reaches $\Pi(v),$ the state can never change using the Coalition Proposal algorithm.
\begin{proposition}
\label{prop:prop9}
    A feasible state $(\mathbf{a},\mathcal{C})$ is absorbing in  the Coalition Proposal algorithm if and only if it belongs to $\Pi(v).$ 
\end{proposition}
\begin{proof}
First, we will prove that if a state $(\mathbf{a},{\mathcal{C}}) \in \Pi(v)$ then it is absorbing. Since $(\mathbf{a},{\mathcal{C}}) \in \Pi(v)$, then $ \sum_{i\in S}a_i \geq v(S) \ $ for all ${S}\subseteq {N},$ hence, no proposal can succeed. Consequently, no player will be able to increase its aspiration nor can any change in the coalitions happen. Furthermore, ${\mathcal{C}}$ is a partition of ${N};$ thus all players are in non-singleton coalitions, or at their individual valuations, and hence, will not decrease their aspiration from the failed proposals. Accordingly, no proposal at this state can change the aspirations or the coalition structure.

Second, the proof of the other direction, a feasible state is absorbing implies that the state is a core solution, follows by contraposition from the proof of Proposition \ref{prop:Empty}. Proposition \ref{prop:Empty} shows that if a state does not belong to the core, then there is always a proposal occurring with positive probability that leads to a change in the environment state. 
\end{proof}
 
In conclusion, the proof of Theorem \ref{thm1} follows from the proof of the existence of a finite steering sequence from any state $(\mathbf{a},{\mathcal{C}})\in \Omega(v)$ to a state $(\mathbf{x},\rho)\in \Pi(v).$ Such sequences of proposals occur at any feasible state with at least some probability $p>0$ that does not depend on the state. By the Borel-Cantelli lemma, the probability that one such sequence is followed at some iteration goes to one as the number of iterations goes to infinity. Finally, Proposition \ref{prop:prop9} shows that any state $(\mathbf{x},\rho)\in \Pi(v)$ is an absorbing state.

\section{AN ILLUSTRATIVE EXAMPLE}
In this section, we apply our algorithm to a multi-agent task allocation setting. We consider self-interested and heterogeneous agents. In task fulfillment problems, it is justifiable to assume that a number of heterogeneous agents are needed to complete a task because of the different resources and capabilities that each agent has \cite{c16}. Self-interested agents aim to maximize their own benefits be it through collaborations or individual actions. In such settings, the core solution provides a satisfactory allocation for said agents. Given a core solution, no group of agents can gain more by deviating from the proposed allocation and coalition structure. In addition, the core solution guarantees the optimal social welfare, the total allocations to the agents equals the maximum welfare value.

\subsection{Setup}
Formulate the multi-agent task allocation problem as a game ${G}=({N},v).$ The set of players ${N}={A}\cup {T},$ where ${A}$ is a set of $m$ agents and ${T}$ is a set of $n$ tasks. $F$ is a set of $k$ features. Each agent is assumed to be equipped with a non-empty subset of these features. Matrix ${Q}$ is an $m\times k$ binary matrix such that $Q_{a,f}=1$ if feature $f$ is present in agent $a$ and $Q_{a,f}=0$ otherwise. For the tasks, ${R}$ is an $n\times k$ binary matrix of task requirements such that $R_{t,f}=1$ if task $t$ requires feature $f$ to be present in the group of agents fulfilling the task and $R_{t,f}=0$ otherwise. The worth of any task $t,$ is characterized by the function $W:T \rightarrow \mathbb{R},$ which is the base value of fulfilling the task. Here, we assume this value to be proportional to the complexity of the task, i.e. the number of features required to fulfill the task. The agents and tasks are set to have a location on a bounded grid. The location is specified by an $(m+n)\times 2$ matrix ${L}.$

 The characteristic function $v(S)$ for $ S\subseteq N$ is formulated as follows: \\
 $v(S)=0$ if 
 \begin{itemize}
      \item $|S|=1, \text{ the coalition is a singleton, or}$
     \item $|T\cap S| \neq 1,$ the coalition does not have exactly one task in it, or
     \item $\min_{f\in F} (\sum_{a \in A \cap S}Q_{a,f}-\sum_{t \in T \cap S}R_{t,f})<0,$ the union of the features that the agents in the coalition have does not cover all of the required features to fulfill the coalition's task. 
 \end{itemize}
Otherwise, \[v(S)=\max\{0, W(t)- \sum_{a \in A \cap S} \lVert L_a-L_t \rVert _1\},\] where $t=T\cap S$ and $L_i$ is the location vector of player $i$.
 
The following are additional assumptions on the setup and the agents' implementation of the algorithm.
\begin{itemize}
    \item We are assuming a static task allocation setup where the agents are required to fulfill only a subset of the available tasks.
    \item The current setup assumes full communication, however, the dynamics are readily applied to limited communication setups so long as agents that are part of any positive-valued coalition can communicate with each other. 
    \item To avoid unnecessary proposals, we assume that agents propose only to form coalitions of positive values, i.e. when an agent proposes a coalition, it makes sure that there is exactly one task in the coalition, the set of features of the agents in the coalition fulfills the task requirements and the agents are close enough to allow for a positive-valued coalition. 
    \item Tasks are assumed to be passive players, they do not propose or have payoff aspirations.
    \item Proposing agents broadcast formed coalitions so that other agents know when tasks in their coalitions leave and cause the dissolution of said coalitions.
\end{itemize}

\subsection{Simulation}
 \emph{Parameters:}
The setup used to produce the succeeding runs is as follows. On a $9\times9$ grid, we randomly generated the $L$ matrix for $10$ players and $20$ tasks. The $Q$ and $R$ matrices were randomly generated accounting for $5$ features.
 \subsubsection{Sample run}
Fig.\ \ref{figure1} shows a visualization of a randomly generated scenario with the above parameters and the coalition structure produced from the attained core solution after running the Coalition Proposal algorithm. Squares represent tasks and circles represent agents. Agents and tasks in the same coalition have the same color and are connected by a dashed line. Grey circles, whenever they exist, are agents that have no benefit of becoming a part of any coalition and grey squares are tasks that were not chosen to be fulfilled. 

Fig.\ \ref{fig2} plots the total aspirations of the players over the communication rounds resulting from running the Coalition Proposal algorithm. The dotted line is the optimal welfare value as solved by a linear program representation of the problem. 

\begin{figure}[thpb]
      \centering
    \includegraphics[scale=0.185]{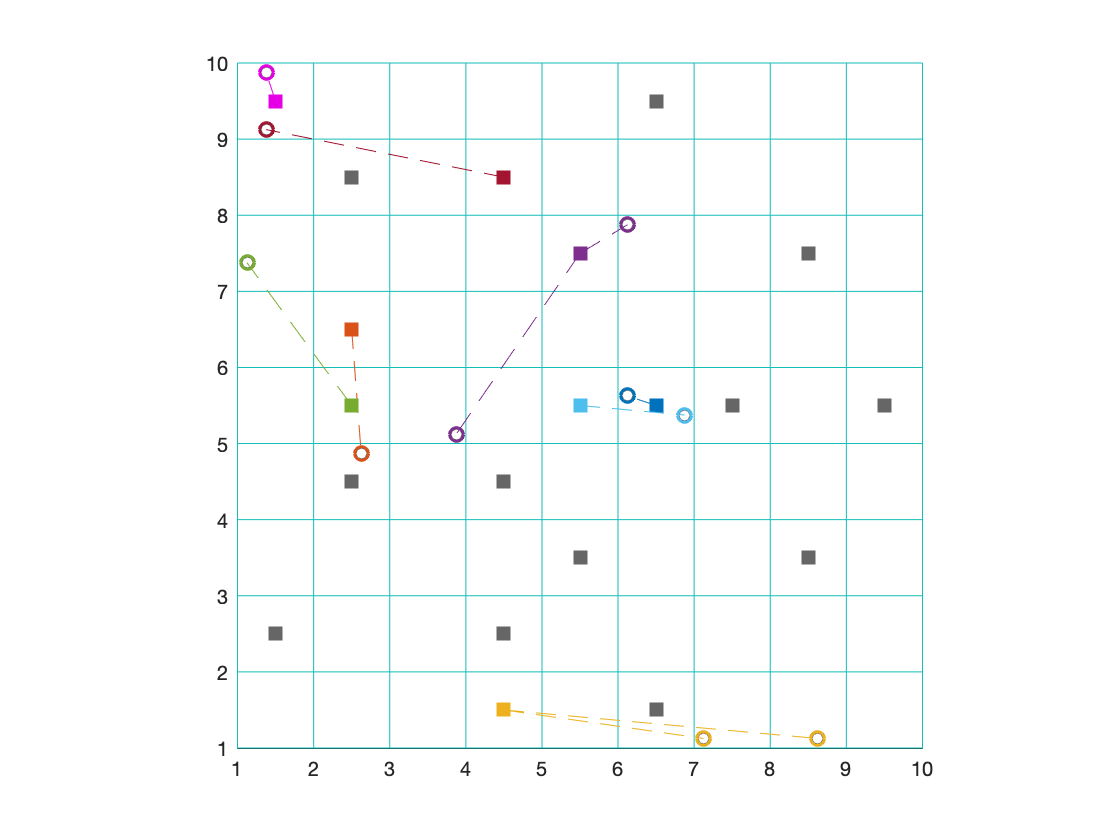}
     
      \caption{A visualization for the tasks and agents distribution over the grid and the formed core coalition structure after running the Coalition Proposal algorithm.}
      \label{figure1}
   \end{figure}
\begin{figure}[thpb]
      \centering
    \includegraphics[scale=0.185]{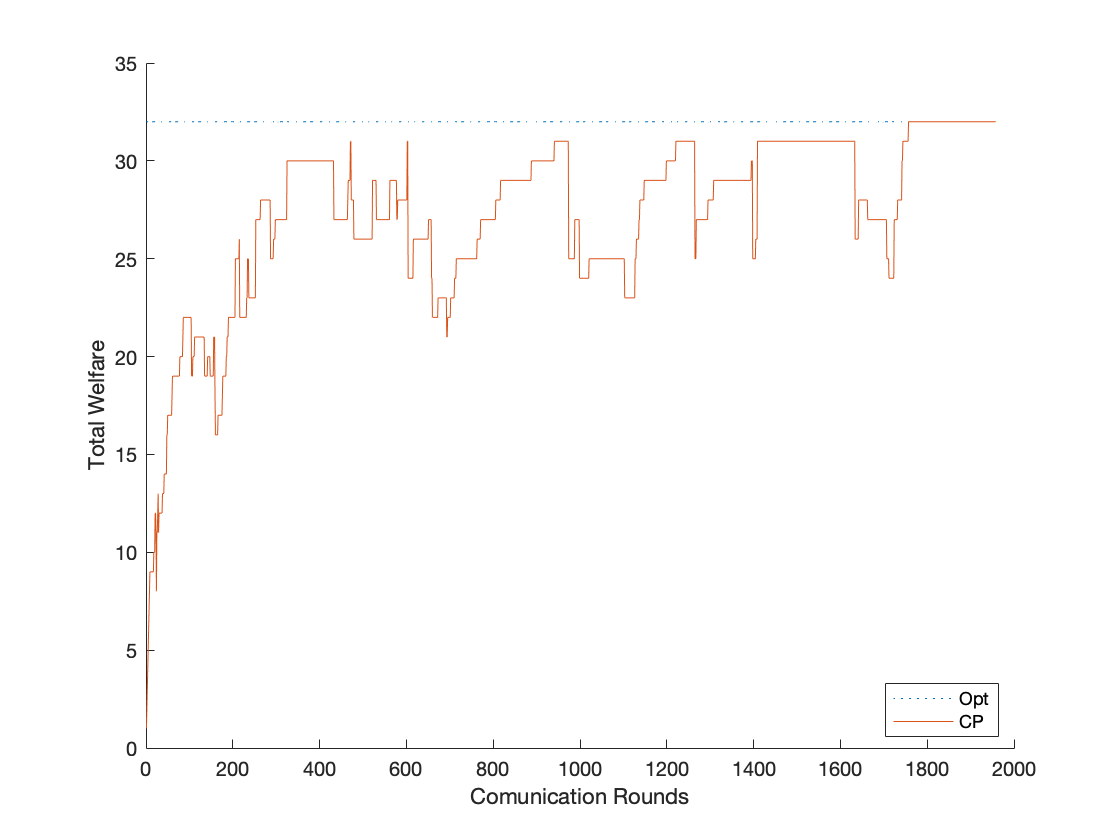} 
     
      \caption{A plot of the total aspirations of the agents obtained from running the Coalition Proposal algorithm, the ``CP" solid line, and the optimal welfare value solved by a linear program, the ``Opt" dotted line.}
      \label{fig2}
   \end{figure}

\subsubsection{Empirical performance for configurations with non-empty core}

 To illustrate the performance of our algorithm in the specified setting. We considered 50 random configurations that produced games with non-empty sets of restricted core solutions. The restricted core solutions are core solutions that allocate zero payoffs to all the tasks. This restriction follows our assumption about the tasks being passive players. We ran our algorithm in addition to three best reply algorithms from \cite{c10} and \cite{c11}. We set the experimentation parameter of the best reply algorithms to $0.05$ and Bernoulli agent activation probability to $0.1$. The best reply algorithm with experimentation was proven to converge to the core in \cite{c10}. The best reply algorithm without experimentation converges very quickly, but possibly, to suboptimal solutions outside the core. For the best reply algorithms with experimentation, an activated agent will not experiment before searching for a coalition, among all of its possible coalitions outside the current coalition structure, that can strictly increase the agent's payoff. This search is computationally expensive. To be able to calculate the best reply, the states of the agents must include the global coalition structure state and the demands of all other agents. 

We ran our proposed algorithm, the best reply algorithm from \cite{c10}, the best reply algorithm with experimentation from \cite{c10}, and the best reply algorithm with experimentation using only feasible demands from \cite{c11} on the aforementioned 50 configurations. Fig.\ \ref{fig3} shows plots of the average of the relative welfare, across the 50 configurations, for the four algorithms. The relative welfare for each configuration is the total aspirations of the players divided by the optimal welfare value.

 \begin{figure}[thpb]
      \centering
  \includegraphics[scale=0.185]{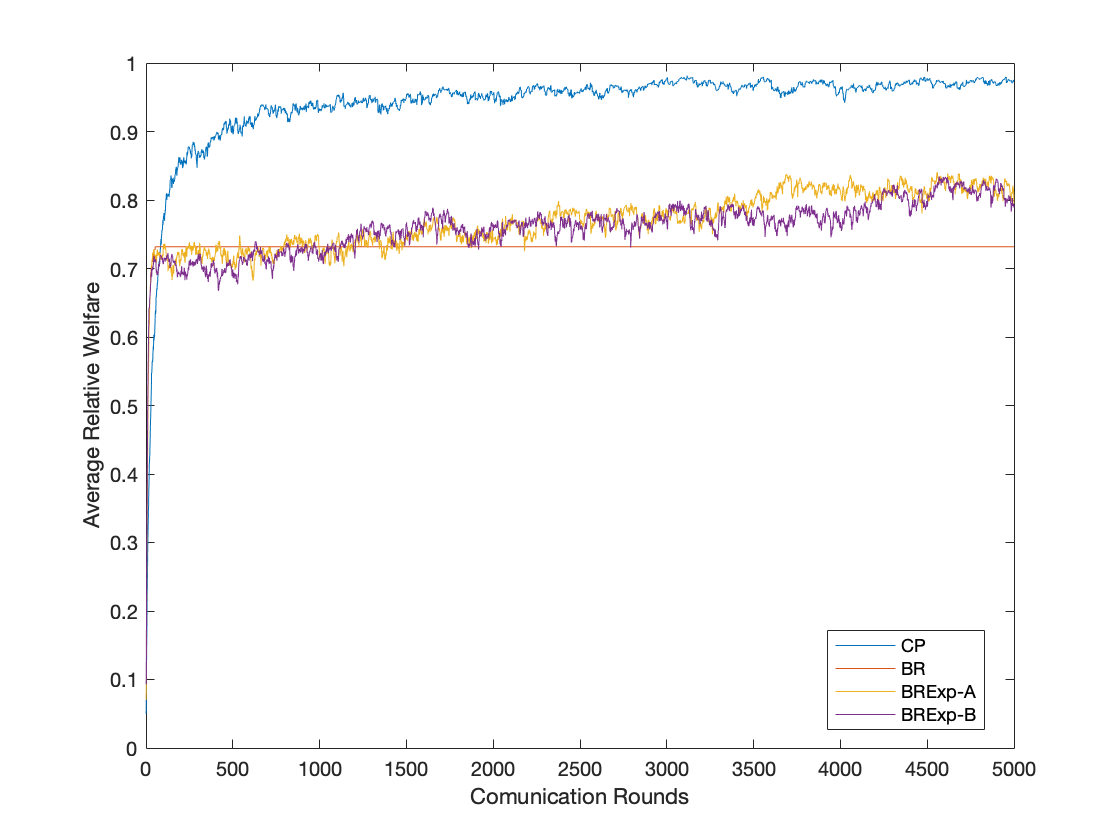}

      \caption{The empirical average performance of the algorithms across 50 configurations with non-empty restricted core. The ``CP" line is for the Coalition Proposal algorithm, the ``BR" line is for the best reply dynamics \cite{c10}, the ``BRExp-A" line is for the best reply dynamics with experimentation \cite{c10}, and the ``BRExp-B" line is for the best reply dynamics with experimentation using only feasible demands \cite{c11}.}
      \label{fig3}
   \end{figure}
\subsubsection{Empirical performance for general configurations}
Using the specified parameters, the configuration resulted in games with non-empty core around half of the time. However, since our algorithm produces a feasible outcome whenever it is terminated, we explored its performance on 50 general configurations, configurations that may produce games with empty core. Even though the algorithm is proven to cycle whenever the core is empty, it may cycle within allocations close to the optimal values. Hence, the algorithm can still produce a good feasible solution if negotiations are set to be terminated after a specified time. Fig.\ \ref{fig4} shows the plots of the average relative welfare using our algorithm in addition to the previously discussed best reply algorithms.

\begin{figure}[thpb]
      \centering
  \includegraphics[scale=0.185]{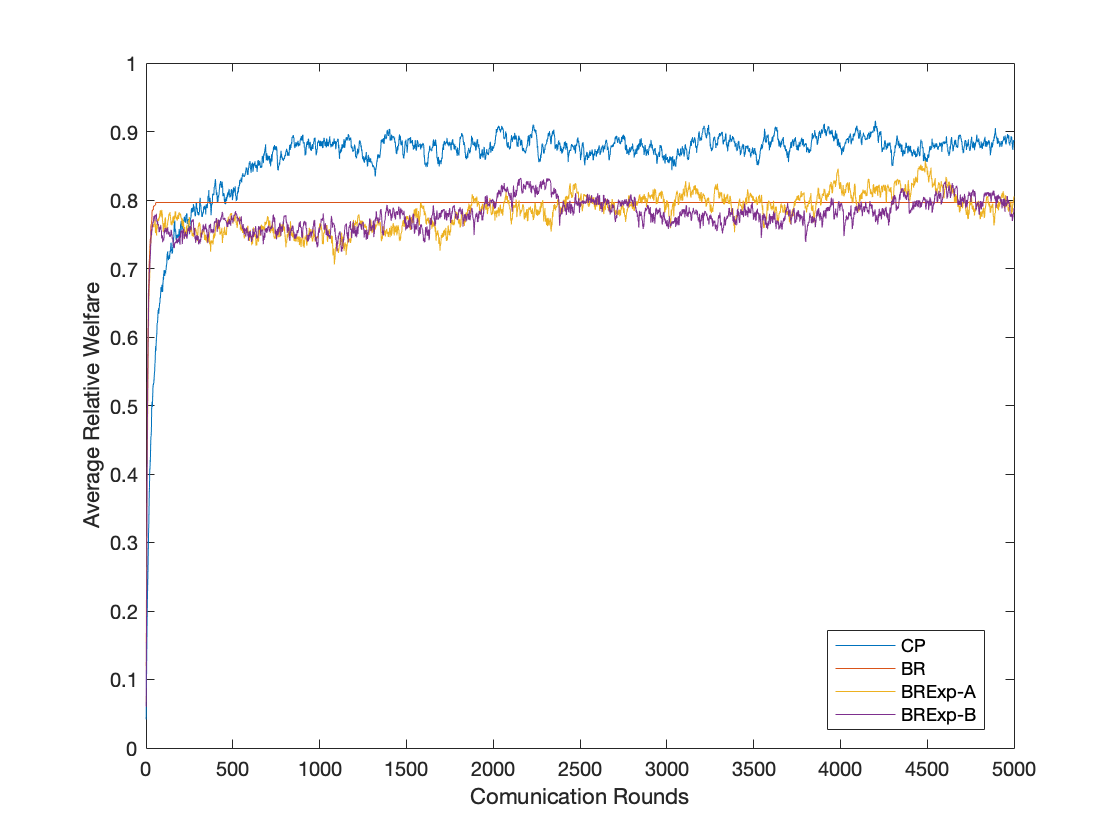}

      \caption{The empirical average performance of the algorithms across 50 general configurations. The ``CP" line is for the Coalition Proposal algorithm, the ``BR" line is for the best reply dynamics \cite{c10}, the ``BRExp-A" is for the best reply dynamics with experimentation \cite{c10}, and the ``BRExp-B" is for the best reply dynamics with experimentation using only feasible demands \cite{c11}. }
      \label{fig4}
   \end{figure}
\subsection{Communication Drops}
An additional consideration in multi-agent systems is communication failures. In many systems, a small percentage of message drops is inevitable. To explore our proposed algorithm's tolerance against communication drops, we have allowed messages that inform agents of the dissolution of a coalition to be dropped with some percentage. Fig. \ref{fig5} illustrates the Coalition Proposal algorithm's performance in the presence of varying levels of communication drops. Even though the message drops affect the convergence guarantee, they allow for faster growth of welfare. Specifically in the initial rounds, a false assumption of being in a coalition dissuades the agents from decreasing their aspirations from failed proposals. Hence, the total welfare gets close to the optimal welfare fast.

 The empirical simulations that we have performed illustrated the advantages of using our algorithm in this multi-agent task allocation setting. Our proposed algorithm allows for negotiation-based distributed decision-making. The algorithm only involves simple computations from one agent in each round of communication and requires limited local knowledge of the environment. Furthermore, we have observed empirically that the algorithm can still reach the optimal welfare value even in the presence of a small percentage of communication failures.

 \begin{figure}[thpb]
      \centering
  \includegraphics[scale=0.185]{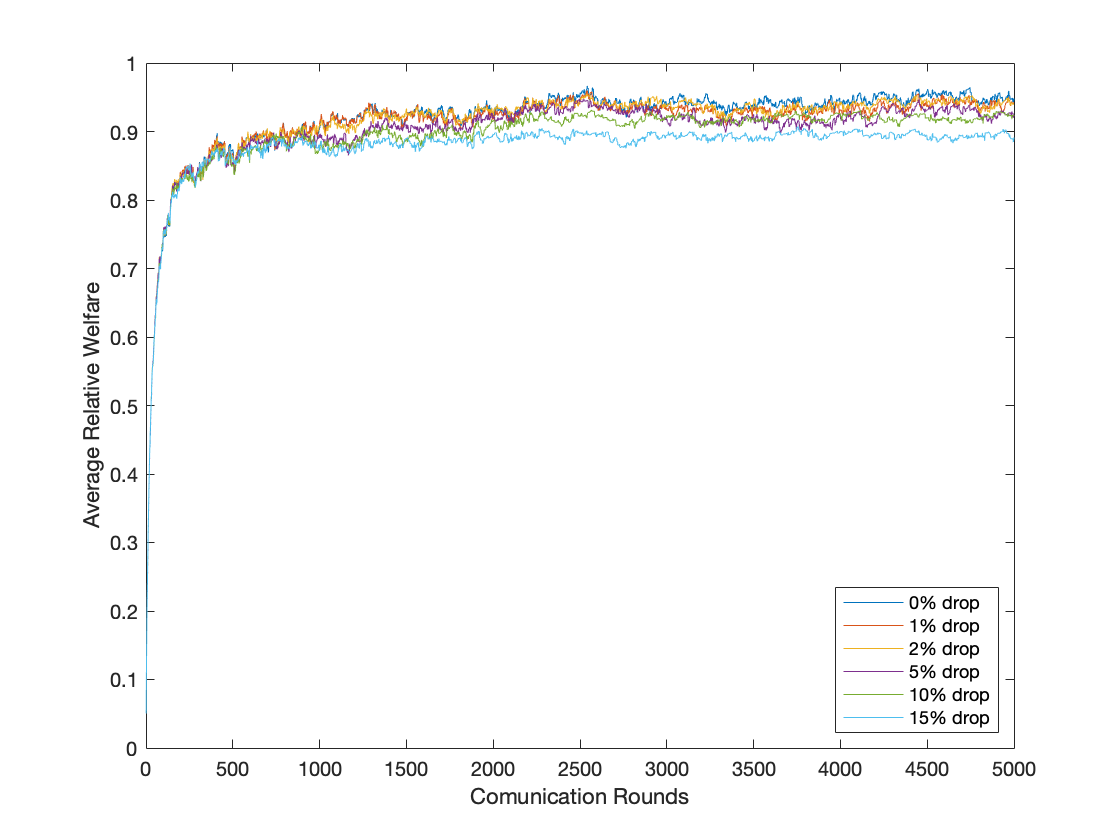}

      \caption{The empirical average performance of the Coalition Proposal algorithm across 50 configurations with non-empty restricted core in the presence of varying levels of communication drops.}
      \label{fig5}
   \end{figure}
 
 \section{CONCLUSION}
 We introduced distributed learning dynamics for coalitional games. We discussed a core solution concept for general TU games. The core solution provides the agents with payoff allocations that preserve individual and coalitional rationality and achieve the optimal social welfare. We proved the convergence of our proposed dynamics to a core solution, whenever one exists. Finally, we illustrated the learning dynamics on a multi-agent task allocation setting and compared it to best reply algorithms and the optimal social welfare value. Our dynamics exhibited desirable performance in simulation for convergence in perfect communication setups as well as in the presence of small percentages of communication failures.

\addtolength{\textheight}{-12cm}   % This command serves to balance the column lengths

%%%%%%%%%%%%%%%%%%%%%%%%%%%%%%%%%%%%%%%%%%%%%%%%%%%%%%%%%%%%%%%%%%%%%%%%%%%%%%%%

\end{document}